\documentclass[conference]{IEEEtran}
\IEEEoverridecommandlockouts
\usepackage{amsmath,amsfonts}
\usepackage{algorithmic}
\usepackage{algorithm}
\usepackage{array}
\usepackage{amsthm}
\usepackage[caption=false,font=normalsize,labelfont=sf,textfont=sf]{subfig}
\usepackage{textcomp}
\usepackage{amssymb}
\usepackage{stfloats}
\usepackage{tikz}
\usepackage{cases}
\usepackage{graphicx}
\usepackage{float,stfloats}
    \graphicspath{{../}}
    \DeclareGraphicsExtensions{.pdf}
\usepackage{color,xcolor}
\usepackage{bm}
\usepackage{arydshln}
\theoremstyle{plain}
\newtheorem{theorem}{Theorem}
\newtheorem{lemma}[theorem]{Lemma}

\newtheorem{claim}[theorem]{Claim}
\theoremstyle{definition}
\newtheorem{example}{Example}
\newtheorem{construction}{Construction}

\usepackage{url}
\usepackage{verbatim}
\usepackage{graphicx}
\usepackage{cite}
\newcommand{\blue}[1]{{\color[rgb]{0.2,0,0.8} #1}}
\newcommand{\red}[1]{{\color[rgb]{0.75,0,0} #1}}
\def\BibTeX{{\rm B\kern-.05em{\sc i\kern-.025em b}\kern-.08em
    T\kern-.1667em\lower.7ex\hbox{E}\kern-.125emX}}
\begin{document}
\title{Bidirectional Piggybacking Design for Systematic Nodes with Sub-Packetization $l=2$}
\author{\IEEEauthorblockN{Ke Wang}
	\IEEEauthorblockA{\textit{IEIT SYSTEMS Co., Ltd., Jinan, 250101, China
} \\
		\textit{KLMM, Academy of Mathematics and Systems Science, Chinese Academy of Sciences, Beijing 100190, China}\\
	Emails: wangke12@ieisystem.com} 
}
\maketitle

\thispagestyle{empty}
\begin{abstract}
	In 2013, Rashmi et~al. proposed the piggybacking design framework to reduce the repair bandwidth 
	of $(n,k;l)$ MDS array codes with small sub-packetization $l$ and it has been studied extensively 
	in recent years.  In this work, we propose an explicit bidirectional piggybacking design (BPD) with
	 sub-packetization $l=2$ and the field size $q=O(n^{\lfloor r/2 \rfloor \!+\!1})$ for systematic 
	 nodes, where  $r=n-k$ equals the redundancy of an $(n,k)$ linear code. And BPD has lower average 
	 repair bandwidth than previous piggybacking designs for $l=2$ when $r\geq 3$.  Surprisingly, 
	 we can prove that the field size $q\leq 256$ is sufficient when $n\leq 15$ and $n-k\leq 4$. 
	 For example, we provide the BPD for the $(14,10)$ Reed-Solomon (RS) code 
	 over $\mathbb{F}_{2^8}$ and obtain approximately $41\%$ savings in the average 
	 repair bandwidth for systematic nodes compared with the trivial repair approach. 
	 This is the lowest repair bandwidth achieved so far for $(14,10)_{256}$ RS codes 
	 with sub-packetization $l=2$.
	
\end{abstract}

\section{Introduction}\label{sec0}
Nowadays, Reed-Solomon (RS)  codes  have attracted great attention in distributed storage systems since they are  Maximum-Distance-Separable (MDS) and thus make optimal use of storage resources for providing reliability.  As in Table \ref{t2}, there  are many companies using RS codes in their storage systems. \begin{table}[htbp]
	\renewcommand\arraystretch{1.3}
	\caption{\scriptsize Some  Reed-Solomon Codes Employed in Distributed Storage Systems}\label{t2}
	\centering 	{\footnotesize \begin{tabular}{|c|c|}
		\hline Storage Systems &$(n,k)$ Reed-Solomon codes\\
		\hline Google File System II (Colossus)& $(9,6)$\\
		\hline Quantcast File System & $(9,6)$\\
		\hline Hadoop 3.0 HDFS-EC & $(9,6)$\\
		\hline Linux RAID-6& $(10,8)$\\
		\hline IBM Spectrum Scale RAID&$(10,8)$, $(11,8)$\\
		\hline Yahoo Cloud Object Store & $(11,8)$\\
		\hline Baidu's Atlas Cloud Storage&$(12,8)$\\
		\hline Facebook's F4 Storage System&$(14,10)$\\
		\hline Microsoft's Pelican Cold Storage & $(18,15)$\\
		\hline
	\end{tabular}}
\end{table}
And the repair of  RS codes has become a significant research topic in recent years.  The repair bandwidth is a important metric in the repairing process\cite{Dimakis2011}.  Obviously, an $(n,k)$ MDS code has a trivial repair approach, i.e., downloading all data from any $k$ nodes and then reconstructing the data stored in the failed node.  Particularly, the ones with the optimal
repair bandwidth are termed as MSR codes \cite{Dimakis2011}. However,  the exponential sub-packetization is necessary for MSR codes. Hence MSR codes are extremely difficult to apply into distributed
storage systems. Consequently,  the main  approaches using in practical are  the  piggybacking designs with small sub-packetization, since piggybacking designs have low repair bandwidth and disk I/O.  

In 2013, Rashmi et al.\cite{Ra2013,Ra2017} proposed the piggybacking design (PD) framework based on the  $(n,k)$  MDS code to reduce the repair bandwidth with small sub-packetization $l$. The failed nodes can be recovered by solving piggybacking functions which are downloaded from surviving nodes. And the piggybacked codes retain the MDS property. And the PDs  were
studied extensively in \cite{Tang2015,Ge2016,Huang2018,Sun2021,Tang2019,Shi2022,Hou2021,Wang2023}. 

For comparisons, we define ARBR as the ratio of the average repair bandwidth of systematic nodes to the bandwidth of the  trivial repair approach. Obviously, ARBR reflects the bandwidth savings in average over the trivial repair approach.
However, for the  sub-packetization $l=2$,  ARBR of PD in  \cite{Ra2013,Ra2017} is still the optimal among all existing PDs.

In order to further reduce the repair bandwidth in PDs, Kralevska $et~al$.\cite{K.k2018} proposed HashTag Erasure Codes (HTECs) with the sub-packetization  $2\leq l\leq r^{\lceil \frac{k}{r} \rceil}$. However the construction of HTECs is not explicit based on the field size $q\geq\binom{n}{r}rl$ where $r=n-k$ and it can not guarantee that the 
 repair of the RS codes shown in Table \ref{t2} can be operated in a field with size $q=256$.

\subsection{Contributions}
In this paper,  we propose a bidirectional piggybacking design (BPD) to reduce the repair bandwidth of systematic nodes with sub-packetization $l=2$.  

Compared with previous PDs, for each substripe $i\in[2]$, the symbols in substripe $j\in[2]\setminus\{i\}$ can be piggybacked to substripe $i$. By contrast, in 
previous PDs only symbols in the substripe 1 are
piggybacked to substripe 2. Therefore, our scheme has more positions to piggyback data symbols and can further reduce ARBR when $l=2$. Although in our scheme, the field size is required to be larger, it still can repair some  major RS codes over $\mathbb{F}_{2^8}$ and we will show that the field size $q\leq 256$ is enough when $n\leq 15$ and $n-k\leq 4$. 

Compared with HTECs\cite{K.k2018}, our scheme is explicit and needs smaller field size, although the  value of the sub-packetization $l$ is more flexible in \cite{K.k2018} as in Table. \ref{t0}. Moreover, in order to achieve the repair bandwidth as low as possible,   the help of computer search is required  in  \cite{K.k2018}.  On the contrary, our scheme is deterministic while retaining ARBR in \cite{K.k2018} when $l=2$. 
\begin{table}[htbp]
	\renewcommand\arraystretch{1.75}
	\caption{\scriptsize Comparison  with HTECs }\label{t0}
	\centering		{\footnotesize	\begin{tabular}{|c|c|c|c|}
		\hline &Sub-packetization $l$&Field size $q$& Explicit \\
		\hline HTECs\cite{K.k2018} &  $2\leq l\leq r^{\lceil \frac{k}{r} \rceil}
		$&$q\geq \binom{n}{r}rl$&No\\
		\hline This paper & $l=2$& $q=O(n^{\lfloor \frac{r}{2} \rfloor \!+\!1})
		$&Yes\\
		\hline 
	\end{tabular}}
	
\end{table}

Specially, in the repair process, sub-packetization $l=2$ means the  sequential I/O. 
This will accelerate the repair speed. For instance, in SATA hard drives, accessing a 4 KB disk block non-sequentially takes milliseconds (or operates at hundreds of KB/s throughput), whereas sequential I/O operations can achieve speeds exceeding 100 MB/s \cite{llc,parac}. Throughout, we focus on the one node failure of  systematic nodes  model.

The remaining of the paper is organized as follows. Section \ref{II}  describes the framework of  the bidirectional  piggybacking design with sub-packetization $l=2$ and show the average repair bandwidth of systematic nodes  of the design.  The explicit construction with the MDS property is proposed in Section \ref{III}. 
Section \ref{sb7} shows that our design is feasible to repair some  RS codes over $\mathbb{F}_{2^8}$.  The comparisons with existing schemes and conclusion are given in Section \ref{sb9}.  

Throughout, we focus on the one node failure of  systematic nodes  model.

\addtolength{\topmargin}{0.03in}

\section{The Bidirectional Piggybacking  Design Framework}\label{II}

In this section, we propose the framework of bidirectional piggybacking design (BPD) with sub-packetization $l=2$.

For any three integers $i,j,t$ where $i<j$, denote by $\left[ i,j \right] =\left\{ i,i+1,\dots \,\,,j \right\} , \left[ i,j \right] +t=\left[i+t,j+t\right]
$ and $\left[ i \right] =\left\{ 1,2,\dots \,\,,i \right\}
$.
\subsection{The bidirectional piggybacking  design}\label{sb0}
\begin{figure}[H]
	\centering
	\includegraphics[width=4.2cm]{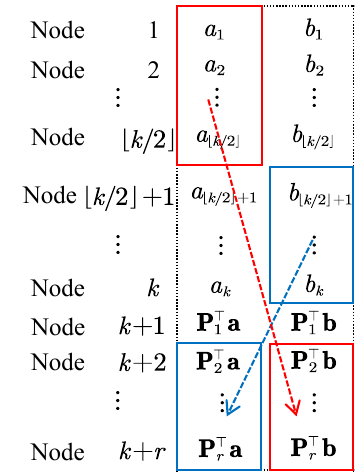}
	
	\caption{\scriptsize The bidirectional piggybacking  design }\label{fg.1}
\end{figure}
For an $\left( n,k;2 \right) $ MDS array code over $\mathbb{F}_q$, the original data symbols $\mathbf{a},\mathbf{b}$  are stored in the $k$ systematic nodes, where $\mathbf{a}=\left( a_{1},\dots ,a_{k} \right) ^\top\in \mathbb{F}_{q}^{k}
$ and $\mathbf{b}=\left( b_{1},\dots ,b_{k} \right) ^\top\in \mathbb{F}_{q}^{k}
$  . The corresponding parity symbols are $\left( \mathbf{P}_{1}^\top\mathbf{a},\dots ,\mathbf{P}_{r}^\top\mathbf{a} \right) ^\top\in \mathbb{F}_{q}^{r}$ and $\left( \mathbf{P}_{1}^\top\mathbf{b},\dots ,\mathbf{P}_{r}^\top\mathbf{b} \right) ^\top\in \mathbb{F}_{q}^{r},$ where $\mathbf{P}_j\in \mathbb{F}_{q}^{k}$, $j\in[r]$. Call  the substripe $\left( a_{1},\dots ,a_{k},\mathbf{P}_{1}^\top\mathbf{a},\dots ,\mathbf{P}_{r}^\top\mathbf{a} \right) ^\top\in \mathbb{F}_{q}^{n} 
$ (also $\left( b_{1},\dots ,b_{k},\mathbf{P}_{1}^\top\mathbf{b},\dots ,\mathbf{P}_{r}^\top\mathbf{b} \right) ^\top\in \mathbb{F}_{q}^{n} 
$ )  the original substripe  throughout this paper.
In  BPD model, as in the original PDs, each original substripe  is an $[n, k]$ MDS codeword.

Denote $k=\lfloor \frac{k}{2} \rfloor +\lceil \frac{k}{2} \rceil$. Let $\alpha_1=\lfloor \frac{k}{2} \rfloor$ and $\alpha_2=\lceil \frac{k}{2} \rceil$. As in Fig. 1, the steps of the bidirectional piggybacking design of data symbols are as follows:

\begin{itemize}
	\item[(1)]   Partition the first $\alpha_1$ entries $\{ a_1,a_2,\dots ,a_{\alpha_1} \} $ of $\mathbf{a}$ into $r-1$ parts evenly, i.e. each part contains $\lfloor \frac{\alpha _1}{r-1} \rfloor$  or $\lceil \frac{\alpha _1}{r-1} \rceil $ symbols. Then add each part to $\left\{ \mathbf{P}_{2}^\top\mathbf{b},\dots ,\mathbf{P}_{r}^\top\mathbf{b} \right\}$ in turn, i.e., add each symbol directly to the corresponding position.
	\item[(2)] Partition the last $\alpha_2$ entries $\left\{ b_{\alpha_1\!+\!{1}},b_{\alpha_1\!+\!{2}},\dots ,b_k \right\} $ of $\mathbf{b}$ into $r-1$ even parts,   i.e. each part contains $\lfloor \frac{\alpha _2}{r-1} \rfloor$  or $\lceil \frac{\alpha _2}{r-1} \rceil $ symbols. In order to retain the MDS property of BPD codes, we need add each part to  $\left\{ \mathbf{P}_{2}^\top\mathbf{a},\dots ,\mathbf{P}_{r}^\top\mathbf{a} \right\}$ in turn by some special functions, such as multiplying each symbol by a variable $\lambda\in \mathbb{F}_q$, then adding it to the corresponding position.
\end{itemize}





\subsection{A toy example}\label{sb1}
\begin{figure}[ht]
	\renewcommand\arraystretch{1.6}
	\centering
	{\footnotesize\begin{tabular}{r|c|c|}
			\cline{2-3}{\rm Node~$1$}&$a_1$&$b_1$\\
			\cdashline{2-3}[0.5pt/1.5pt]$\vdots$&$\vdots$&$\vdots$\\
			\cdashline{2-3}[0.5pt/1.5pt]{\rm Node~$6$}&$a_6$&$b_6$\\
			\cdashline{2-3}[0.5pt/1.5pt]{\rm Node~$7$}&${\mathbf P}_1^\top{\mathbf a}$&${\mathbf P}_1^\top{\mathbf b}$\\
			\cdashline{2-3}[0.5pt/1.5pt]{\rm Node~$8$}&${\mathbf P}_2^\top{\mathbf a}+\blue{\lambda b_4}$&${\mathbf P}_2^\top{\mathbf b}+\red{a_1}$\\
			\cdashline{2-3}[0.5pt/1.5pt]{\rm Node~$9$}&${\mathbf P}_3^\top{\bm a}+\blue{\lambda(b_5+b_6)}$&${\mathbf P}_3^\top{\bm b}+\red{a_2+a_3}$\\\cline{2-3}
	\end{tabular}}
	\caption{\scriptsize A $(9,6;2)$  BPD code}\label{sp1}
\end{figure}

\begin{example}Consider a $(9,6;2)$ BPD code satisfying the MDS property  with a suitable $\lambda\ne0$ in Fig. \ref{sp1}. Clearly, $\alpha_1=\alpha_2=3$. For simplicity , it is sufficient to consider the repair of the first $\alpha_1$ nodes, since the repair process of the first $\alpha_1$ and $\alpha_2$ nodes is similar.
	
If node 1 fails, download $\left\{ b_2,\dots ,b_6,\mathbf{P}_{2}^\top\mathbf{b} \right\} $ first, then the 
symbol $b_1$ can be recovered by the MDS property. For recovering symbol $a_1$, we only need to download $\{\mathbf{P}_{2}^\top\mathbf{b}+a_1\}$, since $\mathbf{P}_{2}^\top\mathbf{b}$ has been recovered in last step and then $a_1$ can be recovered by linear operations. Totally, download $7$ symbols to repair node 1.  If node $2$ or $3$ fails, firstly, $b_2$ or $b_3$ can be recovered through the MDS property with downloading 6 symbols as in recovering symbol $b_1$. Then download $\{\mathbf{P}_{2}^\top\mathbf{b}+a_2+a_3, a_3\}$ or $\{\mathbf{P}_{2}^\top\mathbf{b}+a_2+a_3, a_2\}$  to recover $a_2$ or $a_3$. To repair node $2$ or $3$, need to download 8 symbols. Similarly, one can derive  the repair of nodes $4,5,6$.  Therefore, to repair one node, need to download $\frac{23}{3}$ symbols averagely.  
	
Compared with the trivial repair process which only utilizes the MDS property to repair the failed nodes, it  offers bandwidth	savings of $36.1 \%$.\end{example}

\subsection{The average repair bandwidth ratio of systematic nodes}\label{sb2} 
For an  $\left( n,k;2 \right) $ BPD code, define the repair bandwidth of node $i$ as $\gamma_i$, $i\in[k]$. Then the average repair bandwidth and  repair bandwidth ratio  of systematic nodes are defined as ${\gamma}$ and $\rho$ respectively, where ${\gamma}=\frac{\sum_{i=1}^k{\gamma _i}}{k}, \rho=\frac{{\gamma}}{2k}$.

In $\left( n,k;2 \right) $  BPD design, the systematic nodes are partitioned into two even  parts, i.e. the first $\alpha_1$ and the last $\alpha_2$ nodes. Let $u_1=\lfloor \frac{\alpha_1}{(r-1)} \rfloor, u_2=\lfloor \frac{\alpha_2}{(r-1)} \rfloor, $, and denote $\alpha_1=(r-1)u_1+v_1, \alpha_2=(r-1)u_2+v_2$. 

\begin{theorem}
	The average repair bandwidth ratio  of systematic nodes (ARBR) is 
	$$\rho=\frac{1}{2}+\frac{u_1\alpha_1+u_2\alpha_2+v_1(u_1+1)+v_2(u_2+1)}{2k^2}.$$
\end{theorem}
\begin{proof}
	Since the repair processes of the first $\alpha_1$ nodes and the last $\alpha_2$ nodes are almost the same, we only show the repair of the first $\alpha_1$ nodes.   Let  each of the first $r-1-v_1$ parts of $\{ a_1,a_2,\dots ,a_\frac{k}{2} \} $ contain $u_1$ symbols and the remaining $v_1$ parts each contain $u_1+1$ symbols.
	
	For any $i\in[(r-1-v_1)u_1]$, if node $i$ fails, firstly, download the $k$ symbols $\{b_1,\dots,b_{i-1},b_{i+1},\dots, {\mathbf P}_1^\top{\mathbf b}\}$ to recover the symbol $b_i$ by the MDS property. Secondly, download $u_1$ symbols to recover the symbol $a_i$ by some linear operations as in Section \ref{sb1}. Totally, in order to repair node $i$, we need to download $k+u_1$ symbols. Similarly, for any $j\in[(r-1-v_1)u_1+1,\alpha_1]$, if node $j$ fails, we need download $k+u_1+1$ symbols to repair node $j$. Therefore the total repair bandwidth  of the first $\alpha_1$ nodes is  
	\begin{align}\label{eq1}
		(k+u_1)\alpha_1+v_1(u_1+1).
			\end{align}
Similarly, one can derive the total repair bandwidth  of the last $\alpha_2$ nodes is	\begin{align}\label{eq2}
(k+u_2)\alpha_2+v_2(u_2+1).
\end{align}

	Combining \eqref{eq1} and \eqref{eq2}, it  follows that  $$\rho=\frac{1}{2}+\frac{u_1\alpha_1+u_2\alpha_2+v_1(u_1+1)+v_2(u_2+1)}{2k^2}.$$
\end{proof}

\section{Explicit Construction of BPD }\label{III}

In this section,  we present an explicit construction of $(n,k;2)$ BPD codes over  $\mathbb{F}_q$ with the MDS property where  $q=O(n^{\lfloor \frac{r}{2} \rfloor \!+\!1})$.   

\subsection{ Explicit construction of generator matrices}\label{sb3}
Let $\mathbb{E}$ be a subfield of $\mathbb{F} _q$ with $\left| \mathbb{E} \right|\geq n-1 ,   \left[\mathbb{F} _q:\mathbb{E} \right]=\lfloor \frac{r}{2} \rfloor +1$, thus the field size  $q=O(n^{\lfloor \frac{r}{2} \rfloor \!+\!1})$.  Let $\mathbf{G}=\left( \mathbf{I}\left| \mathbf{P} \right. \right) $  be the generator matrix of each original substripe in an $(n, k; 2)$ BPD code  over $\mathbb{E}$,  where $\mathbf{I}$ is a $k\times k$ identity matrix and $\mathbf{P}=\left( \mathbf{P}_1,\dots ,\mathbf{P}_r \right) $ is a $k\times r$ matrix with  $\mathbf{P}_j=\left( p_{1,j},\dots ,p_{k,j} \right) ^\top\in \mathbb{E} ^k,j\in [r]$.  Since each original substripe is an $[n, k]$ MDS codeword, every square submatrix of  $\mathbf{P}$  is  invertible ($\mathbf{P}$  is superregular).

\begin{construction}\label{cons} Let $\lambda \in \mathbb{F}_q$ satisfying $\mathbb{E}\left( \lambda \right)= \mathbb{F}_q$, then the degree of the minimal polynomial of $\lambda$ over  $\mathbb{E}$ is $ \lfloor \frac{r}{2} \rfloor +1$. As in Section \ref{sb1}, add each symbol of $ \{ a_1,a_2,\dots ,a_{\alpha_1} \} $ directly to the corresponding parity symbol and multiply each symbol  of  $ \{ b_1,b_2,\dots ,b_{\alpha_2} \} $  by  $\lambda$ , then add it to the corresponding position in turn, where $\alpha_1=\lfloor \frac{k}{2} \rfloor$ and $\alpha_2=\lceil \frac{k}{2} \rceil$. \end{construction}

Let  $\tilde{\mathbf{G}}=\left( \mathbf{I}\left| \tilde{\mathbf{P}} \right. \right) $ be the systematic  generator matrix of an  $(n,k;2)$ BPD code  defined above over $\mathbb{F} _q$, where $\mathbf{I}$ is a $2k\times 2k$ identity matrix and $\tilde{\mathbf{P}}$ is a $2k\times 2r$ matrix.   Without loss of generality, hereinafter, we  regard  $\tilde{\mathbf{G}}$ as a $k\times n $ block matrix with each block $ 2\times 2$.  Then  $\tilde{\mathbf{G}}$ can be defined as 
\begin{align}\label{c0}
\small{	\tilde{\mathbf{G}}=\left( \begin{matrix}
		\mathbf{I}_2&		&		&		&		\tilde{\mathbf{P}}_{1,1}&		\tilde{\mathbf{P}}_{1,2}&		\cdots&		\tilde{\mathbf{P}}_{1,r}\\
		&		\mathbf{I}_2&		&		&		\tilde{\mathbf{P}}_{2,1}&		\tilde{\mathbf{P}}_{2,2}&		\cdots&		\tilde{\mathbf{P}}_{2,r}\\
		&		&		\ddots&		&		\vdots&		\vdots&		\vdots&		\vdots\\
		&		&		&		\mathbf{I}_2&		\tilde{\mathbf{P}}_{k,1}&		\tilde{\mathbf{P}}_{k,2}&		\cdots&		\tilde{\mathbf{P}}_{k,r}\\
	\end{matrix} \right)}
\end{align}
where   $\mathbf{I}_2$ is a $2\times 2$ identity matrix and $	\tilde{\mathbf{P}}_{i,j}$ is a $2\times 2$ matrix over $\mathbb{F}_q$, $i\in [k]$ and $j\in [r]$. 
If $2\mid k$ and $r-1\mid \frac{k}{2}$,  by the construction of $\mathbf{G}$, it follows that  
\begin{align}\label{c1}
	\small{\tilde{\mathbf{P}}_{i,j}=\begin{cases}\vspace{0.05cm}
		\left( \begin{matrix}
			p_{i,j}&		1\\
			0&		p_{i,j}\\
		\end{matrix} \right) ,~$if$\,j\in [2,r],i\in \left[ \left( j\!-\!2 \right) u,\left( j\!-\!1 \right) u \right] ,\\
		\vspace{0.05cm}
		\left( \begin{matrix}
			p_{i,j}&		0\\
			\lambda&		p_{i,j}\\
		\end{matrix} \right) ,	~$if$\,j\in [2,r],i\in \left[ \left( j\!-\!2 \right) u,\left( j\!-\!1 \right) u \right] \!+\!\frac{k}{2},\\
		\left( \begin{matrix}
			p_{i,j}&		0\\
			0&		p_{i,j}\\
		\end{matrix} \right) ,~$	otherwise$.\\
	\end{cases}}
\end{align}
where $p_{i,j}\in \mathbb{E}$ corresponding to the entry of $\mathbf{P}_j$ in $\mathbf{G}$ and $\frac{k}{2}=(r-1)u$. It should be pointed out that in all situations, there are only three forms of $\tilde{\mathbf{P}}_{i,j}$ as characterized  in \eqref{c1}.

\begin{example}\label{ex1}
	Consider the generator matrix of a $(9, 6; 2)$ BPD code, then $r-1=2$ and $\frac{k}{2}=3$.  Clearly, $2\nmid 3$.    $\tilde{\mathbf{P}}$ is defined in \eqref{c2}.  The only difference between \eqref{c1} and \eqref{c2} is that neither $ \left( \begin{matrix}
		p_{i,j}&		0\\
		\lambda&		p_{i,j}\\
	\end{matrix} \right) $ nor  $\left( \begin{matrix}
		p_{i,j}&		1\\
		0&		p_{i,j}\\
	\end{matrix} \right) $ does not occur exactly the same times in each column of $[2, r]$.
	\begin{align}\label{c2}\small{\tilde{\mathbf{P}}=\left(\begin{matrix}
			\begin{matrix}
				p_{1,1}&		\\
				&		p_{1,1}\\
			\end{matrix}&		\begin{matrix}
				p_{1,2}&		1\\
				&		p_{1,2}\\
			\end{matrix}&		\begin{matrix}
				p_{1,3}&		\\
				&		p_{1,3}\\
			\end{matrix}\\
			\begin{matrix}
				p_{2,1}&		\\
				&		p_{2,1}\\
			\end{matrix}&		\begin{matrix}
				p_{2,2}&		\\
				&		p_{2,2}\\
			\end{matrix}&		\begin{matrix}
				p_{2,3}&	1	\\
				&		p_{2,3}\\
			\end{matrix}\\
			\begin{matrix}
				p_{3,1}&		\\
				&		p_{3,1}\\
			\end{matrix}&		\begin{matrix}
				p_{3,2}&		\\
				&		p_{3,2}\\
			\end{matrix}&		\begin{matrix}
				p_{3,3}&	1	\\
				&		p_{3,3}\\
			\end{matrix}\\
			\begin{matrix}
				p_{4,1}&		\\
				&		p_{4,1}\\
			\end{matrix}&		\begin{matrix}
				p_{4,2}&		\\
				\lambda&		p_{4,2}\\
			\end{matrix}&		\begin{matrix}
				p_{4,3}&		\\
				&		p_{4,3}\\
			\end{matrix}\\
			\begin{matrix}
				p_{5,1}&		\\
				&		p_{5,1}\\
			\end{matrix}&		\begin{matrix}
				p_{5,2}&		\\
				&		p_{5,2}\\
			\end{matrix}&		\begin{matrix}
				p_{5,3}&		\\
				\lambda&		p_{5,3}\\
			\end{matrix}\\
			\begin{matrix}
				p_{6,1}&		\\
				&		p_{6,1}\\
			\end{matrix}&		\begin{matrix}
				p_{6,2}&		\\
				&		p_{6,2}\\
			\end{matrix}&		\begin{matrix}
				p_{6,3}&		\\
				\lambda	&		p_{6,3}\\
			\end{matrix}\\
		\end{matrix} \right)} \end{align}
\end{example}

\subsection{Proof of  the MDS property}\label{sb5}
Let  $\mathbf{R}$ be a $t\times t$ block submatrix of  $\tilde{\mathbf{P}}$, determined by the row block indices  $i_1<i_2<\cdots <i_t\leq k
$ and the column block indices $j_1<j_2<\cdots <j_t\leq r$.  
Then det($\mathbf{R}$) defines a univariate polynomial $f_\mathbf{R}\left( \lambda \right)\in \mathbb{E}[\lambda]$, since $\left\{ p_{i,j}\left| i\in \left[ k \right] ,j\in \left[ r \right] \right. \right\} \subseteq \mathbb{E} $.

Clearly, if the matrix  defined in \eqref{c0} is the generator matrix of an $(n,k;2)$ MDS array code, then for any $k$ column blocks of  $\tilde{\mathbf{G}}$ is linearly independent over $\mathbb{F}_q$ or equivalently, every square block submatrix $\mathbf{R}$ of $\tilde{\mathbf{P}}$ is invertible, i.e.  $f_\mathbf{R}\left( \lambda \right)\ne 0$.
\begin{theorem}\label{th1}
	An $(n,k;2 )$ BPD code given by Construction \ref{cons}  is MDS over $\mathbb{F}_q$ where $q=O(n^{\lfloor \frac{r}{2} \rfloor \!+\!1})$.
\end{theorem}
Theorem \ref{th1} is the main result of this subsection, which  will be proved at the end of this subsection.
\begin{lemma}\label{le0}
	For any $\mathbf{R},$  $f_\mathbf{R}(0)\ne 0$.
\end{lemma}
\begin{proof}
	If $\lambda=0$, then BPD code  degenerate to the original piggybacking design code, clearly,  it has the MDS property and for any $\mathbf{R}$,  $f_\mathbf{R}(0)\ne 0$. That is, $f_\mathbf{R}(\lambda)$ is a nonzero polynomial.
\end{proof}
\begin{lemma}\label{le1}
	If all non-diagonal   $2\times 2$ block entries of $\mathbf{R}$ are 	upper triangular matrices or lower triangular matrices, then $f_\mathbf{R}\left( \lambda \right)\ne 0$ and $
	\deg\left( f_{\mathbf{R}}\left( \lambda \right) \right)=0 
	$.
\end{lemma}
\begin{proof}
	Since such $\mathbf{R}$ can be seen as a square block submatrix of the parity blocks in a systematic generator  matrix of  a original piggybacking code,  we have  $f_\mathbf{R}\left( \lambda \right)\ne 0$  for any $\lambda$,  it follows that $
	\deg\left( f_{\mathbf{R}}\left( \lambda \right) \right) 
	$=0. Actually, we can obtain $f_\mathbf{R}\left( \lambda \right)=f_\mathbf{R}(0)$.
\end{proof}

Before proving the Theorem \ref{th1}, we show the $(9,6;2)$ BPD code defined in Example \ref{ex1} is MDS.

Throughout this paper, the blank of a matrix denotes zero and $*$ denotes nonzero element.
\begin{claim}
	In Example \ref{ex1},  for any square block submatrix $\mathbf{R}$ of $\tilde{\mathbf{P}}$, $
	\deg\left( f_{\mathbf{R}}\left( \lambda \right) \right) 
	$$\leq 1$ and the $(9,6;2)$ BPD code is MDS.
\end{claim}
\begin{proof}
	By Lemma \ref{le1}, if $
	\deg\left( f_{\mathbf{R}}\left( \lambda \right) \right) 
	$$\geq2$, then $\mathbf{R}$ must be of the form  
	\begin{align}\label{c3}
	\small{	\left(\begin{matrix}
			\begin{array}{cc:cc:cc}
			\ast&		&		\ast&		x&		\ast&		y\\
			&		\ast&		&		\ast&		&		\ast\\\hdashline
			\ast&		&		\ast&		&		\ast&		\\
			&		\ast&		\lambda&		\ast&		&		\ast\\\hdashline
			\ast&		&		\ast&		&		\ast&		\\
			&		\ast&		&		\ast&		\lambda&		\ast\\
			\end{array}
		\end{matrix} \right) }, \mathrm{where} ~xy=0 ~\mathrm{and}~ x,y\in \left\{ 0,1 \right\},\notag
	\end{align} and the coefficient of $\lambda^2$ is $\det(\mathbf{D}) $ where $\small{\mathbf{D}=\left( \begin{matrix}
		\ast&		&		x&		y\\
		&		\ast&		\ast&		\ast\\
		\ast&		&		&		\\
		\ast&		&		&		\\
	\end{matrix} \right) }$. Since the last two rows of  $\mathbf{D}$  are linearly dependent over $\mathbb{F}_q$,  it follows that $\det(\mathbf{D})=0 $ and  $
	\deg\left( f_{\mathbf{R}}\left( \lambda \right) \right) 
	$$<2$ . 
	
	Moreover,  suppose that  there exists $\mathbf{R}$  such that $f_\mathbf{R}\left( \lambda \right)$=0.  There are two cases of $f_\mathbf{R}\left( \lambda \right)$:
	
	$\textbf {Case~1:}$ $
	\deg\left( f_{\mathbf{R}}\left( \lambda \right) \right) 
	$=0. Then  $f_\mathbf{R}\left( \lambda \right)=f_\mathbf{R}(0)\ne 0$ by Lemma \ref{le0}.
	
	$\textbf {Case~2:}$ $
	\deg\left( f_{\mathbf{R}}\left( \lambda \right) \right) 
	$=1.  If  $f_\mathbf{R}\left( \lambda \right)=0$, which  contradicts with the degree of the minimal polynomial of $\lambda$ over  $\mathbb{E}$ is 2, since  $f_\mathbf{R}\left( \lambda \right)\in \mathbb{E}[\lambda]$ and $\left[\mathbb{F} _q:\mathbb{E} \right]=\lfloor \frac{r}{2} \rfloor +1=2$,  $\mathbb{E}\left( \lambda \right)= \mathbb{F}_q$. 
	
	Therefore we have $
	\deg\left( f_{\mathbf{R}}\left( \lambda \right) \right) 
	$$<2$ and the $(9,6,2)$ BPD code is MDS.
\end{proof}

\begin{lemma}\label{ree}
	Let  $\tilde{\mathbf{G}}=\left( \mathbf{I}\left| \tilde{\mathbf{P}} \right. \right) $ be the systematic  generator matrix of an  $(n,k;2)$ BPD code  defined in Construction \ref{cons}, then for any  square block submatrix $\mathbf{R}$ of $\tilde{\mathbf{P}}$, $\deg(f_\mathbf{R}\left( \lambda \right))$$\leq \lfloor \frac{r}{2} \rfloor $.
\end{lemma}
\begin{proof}
	Consider the construction of $\mathbf{R}$.  As in \eqref{c1},  the $2\times 2$ block entries of  the first $\lfloor\frac{k}{2}\rfloor$  row blocks of  $\tilde{\mathbf{P}}$ are all diagonal   or	upper triangular matrices; and the entries of  the remaining  $\lceil\frac{k}{2}\rceil$  row blocks of  $\tilde{\mathbf{P}}$ are all diagonal  or lower triangular matrices.  Furthermore, there is only one upper triangular or lower triangular matrix in each row block, since any one data symbol  can't be piggybacked  twice in BPD. Moreover, since $\mathbf{R}$ is a $t\times t$  block submatrix of  $\tilde{\mathbf{P}}$, we may assume 
	the first $g$  and  the remaining $h$  row blocks of 
	$\mathbf{R}$ are taken from the first $\lfloor\frac{k}{2}\rfloor$ and  the last $\lceil\frac{k}{2}\rceil$  row blocks of  $\tilde{\mathbf{P}}$ respectively, where $g+h=t, t\leq r$.
	
	For  the fixed $t,g,h$, we have $\deg\left( f_{\mathbf{R}}\left( \lambda \right) \right)  \leq h$. 
	Suppose  $
	\deg\left( f_{\mathbf{R}}\left( \lambda \right) \right) 
	$$\geq \lfloor \frac{r}{2} \rfloor+1 $, then $h\geq \lfloor \frac{r}{2} \rfloor+1$. Let $h'\in\left [ \lfloor \frac{r}{2} \rfloor+1, h\right]$ and $
	\deg\left( f_{\mathbf{R}}\left( \lambda \right) \right) 
	$=$h'$, then the block matrices of the form $\begin{pmatrix}
		*&\\
		\lambda &*\\
	\end{pmatrix}$ are  distributed in at least $h'$ different column blocks in $\mathbf{R}$. We may assume there exist  $h'$ row blocks that are of the following form 
$$
\small{\left( \begin{matrix}
	\begin{array}{cc:cc:cc:cc:c}
	*&		&		*&		&		*&		&		*&		&		\cdots\\
	&		*&		\lambda&		*&		&		*&		&		*&		\cdots\\\hdashline
	*&		&		*&		&		*&		&		*&		&		\cdots\\
	&		*&		&		*&		\lambda&		*&		&		*&		\cdots\\\hdashline
	\vdots&		\vdots&		\vdots&		\vdots&		\vdots&		\vdots&		\vdots&		\vdots&		\cdots\\\hdashline
	*&		&		*&		&		*&		&		*&		&		\cdots\\
	&		*&		&		*&		&		*&		\lambda&		*&		\cdots\\
	\end{array}
\end{matrix} \right),} 
$$where  $\small{\begin{pmatrix}
	*&\\
	\lambda &*\\
\end{pmatrix}}$ are  distributed in exactly $h'$ different column blocks.

Hence if $
\deg\left( f_{\mathbf{R}}\left( \lambda \right) \right) 
$=$h'$, there  must exists a matrix $\mathbf{D}_{h'}$ satisfying $\det(\mathbf{D}_{h'})\ne 0$, where $\mathbf{D}_{h'}$ is a $(2t-h')\times (2t-h')$ matrix over  $\mathbb{F}_q$ and $\mathbf{D}_{h'}$ is obtained from $\mathbf{R}$ by removing the $h'$ columns and $h'$ rows where $\lambda$ is located at. Then the form of  $\mathbf{D}_{h'}$ is equivalent to  $\small{\left( \begin{matrix}
		\mathbf{A}_{2(t-h')\times \left( t-h' \right)}&		\mathbf{B}_{2(t-h')\times t }	\\
		\mathbf{C}_{ \left( 2h'-h' \right)\times \left( t-h' \right)}&			\mathbf{0}_{ \left( 2h'-h' \right)\times t}\\
	\end{matrix} \right)} $, i.e, we can do  elementary row and column transformations to $\mathbf{D}_{h'}$, where $\mathbf{A}_{i,j}$ (also $\mathbf{B}_{i,j}$ and $\mathbf{C}_{i,j}$) denotes a $i\times j$ matrix over $\mathbb{F}_q$ .   Since $t\leq r$ and $h'\geq\lfloor \frac{r}{2} \rfloor+1 $, we have $h'>t-h'$, it follows that  rank$\left(\left( \begin{matrix}\mathbf{C}_{h'\times \left( t-h' \right)}&		\mathbf{0}_{h'\times t}\\
	\end{matrix} \right) \right)\leq t-h'<h'$ over $\mathbb{F}_q$. Then $\mathbf{D}_{h'}$ is not of full rank and $\det(\mathbf{D}_{h'})=0$.  Therefore $
	\deg\left( f_{\mathbf{R}}\left( \lambda \right) \right) 
	$$\leq \lfloor \frac{r}{2} \rfloor $.  
\end{proof}

\begin{proof}[Proof of Theorem 2]
	It is sufficient to prove for any $\mathbf{R}$, $f_{\mathbf{R}}\left( \lambda \right)\ne 0$. There are two cases of  $f_{\mathbf{R}}\left( \lambda \right)$:
	
	\textbf{Case  1:} $\deg\left( f_{\mathbf{R}}\left( \lambda \right) \right) =0$, i.e. $f_{\mathbf{R}}\left( \lambda \right)=f_{\mathbf{R}}(0)\ne 0$ through Lemma \ref{le0}.
	
	\textbf{Case 2:} $\deg\left( f_{\mathbf{R}}\left( \lambda \right) \right) \geq0$,  by Lemma \ref{ree},  $\deg\left( f_{\mathbf{R}}\left( \lambda \right) \right) \leq \lfloor \frac{r}{2} \rfloor $, and by the definition of $\lambda$, the degree of the minimal polynomial  of $\lambda$ over $\mathbb{E}$
	is $ \lfloor \frac{r}{2} \rfloor+1$, therefore $f_{\mathbf{R}}\left( \lambda \right) \ne 0.$
\end{proof}

\section{Practical Considerations}\label{sb7}
In this section, we will show that our scheme is feasible to repair  Reed Solomon codes over $\mathbb{F}_{2^8}$  with $n\leq16$  when  $2\leq r\leq 3$\ and   $n\leq15$ when $r=4$. 

For RS codes, we need the ground field $\left| \mathbb{E} \right|\geq n$ and we choose $\mathbb{E} =\mathbb{F} _{2^4}$ in this section. Since the generator matrix of the original substripes is defined in $\mathbb{E} =\mathbb{F} _{2^4}$, the evaluation points set of  a Reed Solomon code here is  belong to $\mathbb{F}_{2^4}$.

If $r=2$ or $3$, we have $\lfloor\frac{r}{2} \rfloor+1=2$, then the degree of the minimal polynomial of $\lambda$ is 2, it follows that   $\mathbb{F} _{2^4}\left( \lambda \right) =\mathbb{F} _{2^8}$.  Therefore  for any $\lambda \in \mathbb{F} _{2^8}\backslash \mathbb{F} _{2^4}$,  the $(n,k;2)$ BPD codes are MDS over $\mathbb{F}_{2^8}$.

If  $r=4$, then  $\lfloor\frac{r}{2} \rfloor+1=3$,  as in Construction \ref{cons}, we need $\mathbb{F}_q=\mathbb{F}_{2^{12}}$. However,  we have the following theorem:
\addtolength{\topmargin}{0.03in}
\begin{theorem}\label{thm7}
	If  $r=4$ and $4<n\leq 15$, there exists $\lambda \in \mathbb{F}_{2^8}$ with $\mathbb{F}_{2^4}{(\lambda)}=\mathbb{F}_{2^8}$ such that the $(n,k;2)$ BPD codes defined in Construction \ref{cons} are  MDS.
\end{theorem}
\begin{proof}
	We choose $\mathbb{E}=\mathbb{F}_{2^4}$. As in Section \ref{sb3},  let   $\tilde{\mathbf{G}}=\left(\mathbf{I}\left| \tilde{\mathbf{P}} \right. \right) $ be the systematic  generator matrix of an  $(n,k;2)$ BPD code and $\mathbf{R}$ be a $t\times t$  block submatrix of $\tilde{\mathbf{P}}$ where $t\in [4]$.
	Thus it is sufficient to prove that for any ${\mathbf{R}}$, $\exists \lambda \in \mathbb{F}_{2^8} $  such that  $\det(\mathbf{R})=f_{\mathbf{R}}({\lambda})\ne 0$. 
	
	By Lemma \ref{ree}, we have $deg\left( f_{\mathbf{R}}\left( \lambda \right) \right) \leq 2
	$. 	Since the degree of the  minimal polynomial of  $\lambda$ over  $\mathbf{F}_{2^4}$ is 2,  we only consider the determinant of these $\mathbf{R}$  with $deg\left( f_{\mathbf{R}}\left( \lambda \right) \right) =2$. By Lemma \ref{le1}, the $2\times 2$ block entries of  $\mathbf{R}$ must contain upper and lower triangular matrices.  Then the form of $\mathbf{R}$ must be equivalent to the form  in the following  
	$$
		\small{\mathbf{R}_0=\begin{pmatrix}
			\begin{array}{cc:cc:cc:cc}
			\ast&		&		\ast&		1&		\ast&	&		\ast&	\\
			&		\ast&		&		\ast&		&		\ast&		&		\ast\\\hdashline
			\ast&		&		\ast&		&		\ast&		1&		\ast&	\\
			&		\ast&		&		\ast&		&		\ast&		&		\ast\\\hdashline
			\ast&		&		\ast&		&		\ast&		&		\ast&		\\
			&		\ast&		\lambda&		\ast&		&		\ast&		&		\ast\\\hdashline
			\ast&		&		\ast&		&		\ast&		&		\ast&		\\
			&		\ast&		&		\ast&		\lambda&		\ast&		&		\ast\\
			\end{array}
		\end{pmatrix} } ,
$$and let  these $\mathbf{R}_0$  form a set $\mathbf {M}_n$.
	Since the following four  cases can be excluded:
	
	\textbf{Case 1:}
$$\small{\mathbf{R}_1=\begin{pmatrix}
	\begin{array}{cc:cc:cc}
		\ast&		&		\ast&	1&		\ast&			\\
		&		\ast&		&		\ast&		&		\ast	\\\hdashline
		\ast&		&		\ast&		&		\ast&		\\
		&		\ast&		\lambda&		\ast&		&		\ast\\\hdashline
		\ast&		&		\ast&		&		\ast&		\\
		&		\ast&		&		\ast&		\lambda&		\ast\\
\end{array}		\end{pmatrix}, }	$$  if  $deg\left( f_{\mathbf{R}}\left( \lambda \right) \right) =2$,  then the coefficient of $\lambda^2$ is $\det \left( \mathbf{D}_1 \right) $, where $\small{\mathbf{D}_1= \begin{pmatrix}
		\ast&		&		1&		\\
		&		\ast&		&		\ast\\
		\ast&		&		&		\\
		\ast&		&		&		\\
	\end{pmatrix} }$, clearly, the last two rows of  $\mathbf{D_1}$ are  linearly dependent over $\mathbb{F}_{2^8}$ and $\det \left( \mathbf{D}_1 \right) =0$ which contradicts with  $deg\left( f_{\mathbf{R}}\left( \lambda \right) \right) =2$.
	
	\textbf{Case 2:}
	$$
	\mathbf{R}_2=\small{\begin{pmatrix}
		\begin{array}{cc:cc:cc:cc}
			\ast&		&		\ast&	1&		\ast&	&		\ast&		\\
		&		\ast&		&		\ast&		&		\ast&		&		\ast\\\hdashline
		\ast&		&		\ast&	1&		\ast&		&		\ast&	\\
		&		\ast&		&		\ast&		&		\ast&		&		\ast\\\hdashline
		\ast&		&		\ast&		&		\ast&		&		\ast&		\\
		&		\ast&		\lambda&		\ast&		&		\ast&		&		\ast\\\hdashline
		\ast&		&		\ast&		&		\ast&		&		\ast&		\\
		&		\ast&		&		\ast&		\lambda&		\ast&		&		\ast\\
		\end{array}
	\end{pmatrix}  }, 
	$$
	similarly, the coefficient of  $\lambda^2$ is  $\det \left( \mathbf{D}_2 \right) $, where $\small{\mathbf{D}_2=\left(\begin{matrix}
		\begin{array}{cc:cc:cc}
		\ast&		&		1&		&		\ast&		\\
		&		\ast&		\ast&		\ast&		&		\ast\\\hdashline
		\ast&		&		1&		&		\ast&		\\
		&		\ast&		\ast&		\ast&		&		\ast\\\hdashline
		\ast&		&		&		&		\ast&		\\
		\ast&		&		&		&		\ast&		\\
		\end{array}
	\end{matrix} \right) }$.  Since the four rows  $\small{\left( \begin{matrix}
		\ast&		&		1&		&		\ast&		\\
		\ast&		&		1&		&		\ast&		\\
		\ast&		&		&		&		\ast&		\\
		\ast&		&		&		&		\ast&		\\
	\end{matrix} \right)} $ of  $\mathbf{D}_2$ are linearly dependent, it follows that   $\det \left( \mathbf{D}_2 \right)=0 $.
	
Due to the lack of space, we discuss the remaining two cases in the Appendix.
	
	Moreover,  $4<n\leq 15$ and it is sufficient to consider the field size when  $n=15$, since the number of the block submatrices in $\mathbf{M}_{15}$ of  $\tilde{\mathbf{P}}$   is  the largest when $n=15$ and  $k=n-r=11$.	The matrix $\tilde{\mathbf{P}}$ is composed of  two parts: the first $5$ row blocks and the last $6$ row blocks. To obtain a matrix of the form like $\mathbf{R}_0$, we need choose two row blocks from each of the two parts. Clearly, if the two row blocks chosen from the first part are equivalent to the following form  can be ruled out:
		 $$\small \begin{pmatrix}
			\begin{array}{cc:cc:cc:cc}
				\ast&		&		\ast&	1&		\ast&	&		\ast&		\\
				&		\ast&		&		\ast&		&		\ast&		&		\ast\\\hdashline
				\ast&		&		\ast&	1&		\ast&		&		\ast&	\\
				&		\ast&		&		\ast&		&		\ast&		&		\ast\\
				\end{array}
			\end{pmatrix}.$$ Similarly, if the two row blocks chosen from the last part are equivalent to the following form  can be ruled out:
		 $$\small \begin{pmatrix}
			\begin{array}{cc:cc:cc:cc}
				\ast&		&		\ast&	&		\ast&	&		\ast&		\\
				&		\ast&	\lambda	&		\ast&		&		\ast&		&		\ast\\\hdashline
				\ast&		&		\ast&	&		\ast&		&		\ast&	\\
				&		\ast&	\lambda	&		\ast&		&		\ast&		&		\ast\\
			\end{array}
		\end{pmatrix}.$$ Then by some simple calculations, we obtain there are $\left(\binom{5}{2}-2\right)\times \left(\binom{6}{2}-3\right) =96$ such $\mathbf{R}_0$, since some cases can be excluded.    
	
	Suppose that for any $\mathbf{R}_0 \in  \mathbf{M}_{15}$,  $\deg\left( f_{\mathbf{R}_0}\left( \lambda \right) \right) =2$,  then there are at most $96\times 2 =192$ elements of $\mathbb{F}_q\backslash \mathbb{E} $ such that $
	\prod_{\mathbf{R}_0\in \mathbf{M}_{15}}{f_{\mathbf{R}_0}\left( \lambda \right)}=0$. Since  $\left| \mathbb{F} _{2^8}\backslash \mathbb{F} _{2^4} \right|=240>192$,  there exists such  $\lambda\in \mathbb{F} _{2^8}\backslash \mathbb{F} _{2^4}$ such that  the  $(15,11,2)$ BPD is MDS.  
	Thus the proof is completed.
\end{proof}
\section{Comparisons and Conclusion}\label{sb9}
\subsection{Comparisons}
In this section, we will make  comparisons of  BPD with PDs. 

Since the subsequent PDs\cite{Tang2015,Ge2016,Huang2018,Tang2019,Sun2021,Hou2021,Shi2022} do not outperform \cite{Ra2013,Ra2017} in repairing MDS array codes when $l=2$, we only compare the average repair bandwidth ratio of systematic nodes with \cite{Ra2013,Ra2017} in Table \ref{t1}. 

\begin{table}[htbp]
	\renewcommand\arraystretch{1.2}
	\begin{center}
		\caption{\scriptsize ARBR of  Some Common Eemployed  RS Codes Over $\mathbb{F}_{2^8}$ }\label{t1}
		\centering
		\footnotesize{	\begin{tabular}{|c|c|c|c|c|}
			\hline RS$(n,k)$ &K. V. Rashmi $et~ al$.\cite{Ra2013,Ra2017}& This paper\\
			\hline  $\left(9,6\right)$ &$69.4\%$&$63.9\%$\\
			\hline  $(11,8)$&$68.75\%$& $62.5\%$\\		
			\hline$(12,8)$ & $65.6\%$&$59.4\%$\\
			\hline $(14,10)$ & $67.5\%$&$59\%$\\			
			\hline 		
		\end{tabular}}
	\end{center}
\end{table}	

Actually, as $r$ goes larger, the  advantage in ARBR of BPD  becomes more evident.  The reason we only show the codes with $r\leq 4$ is that the redundancy $r$ of the codes currently used in distributed storage systems is usually small.





\subsection{Conclusion}\label{V}
In this paper, we propose the bidirectional piggybacking design (BPD) with $l=2$ which can further reduce the repair bandwidth of systematic nodes of the  piggybacking design in \cite{Ra2013,Ra2017} with $l=2$, however we need larger field size to maintain the MDS property. How to further reduce the field size is our future work.

\appendix[The remaing two cases of Theorem \ref{thm7}]
		\textbf{Case 3:}
	
	$$ \mathbf{R}_3=\begin{pmatrix}
		\begin{array}{cc:cc:cc:cc}
			\ast&		&		\ast&	1&		\ast&	&		\ast&		\\
			&		\ast&		&		\ast&		&		\ast&		&		\ast\\ \hdashline
			\ast&		&		\ast&	&		\ast&		&		\ast&	\\
			&		\ast&		\lambda	&		\ast&		&		\ast&		&		\ast\\ \hdashline
			\ast&		&		\ast&		&		\ast&		&		\ast&		\\
			&		\ast&	&		\ast&		\lambda	&		\ast&		&		\ast\\ \hdashline
			\ast&		&		\ast&		&		\ast&		&		\ast&		\\
			&		\ast&		&		\ast&		&		\ast&	\lambda	&		\ast\\
		\end{array}
	\end{pmatrix}  , $$
	if    $\deg\left( f_{\mathbf{R}}\left( \lambda \right) \right) =2$,  then there must exists  a  matrix $\mathbf{D}_3$ of the form $\left( \begin{matrix}
		\begin{array}{cc:cc:cc}
		\ast&		&		1&		&		\ast&		\\
		&		\ast&		\ast&		\ast&		&		\ast\\\hdashline
		\ast&		&		&		&		\ast&		\\
		\ast&		&		&		&		\ast&		\\\hdashline
		\ast&		&		&		&		\ast&		\\
		&		\ast&		\ast&		\ast&		\mathrm{\lambda}&		\ast\\
		\end{array}
	\end{matrix} \right) $ with  $\det \left( \mathbf{D}_3 \right)\ne 0 $.  However,  the three rows  $\left( \begin{matrix}
		
		\ast&		&		&		&		\ast&		\\
		\ast&		&		&		&		\ast&		\\
		\ast&		&		&		&		\ast&		\\
		
	\end{matrix} \right) $ of $\mathbf{D}_3$ are linearly dependent, we have $\det \left( \mathbf{D}_3 \right)= 0 $.
	
	\textbf{Case 4:}
	$$
		\mathbf{R}_4=\left(\begin{matrix}
			\begin{array}{cc:cc:cc:cc}
			\ast&		&		\ast&	1&		\ast&	&		\ast&		\\
			&		\ast&		&		\ast&		&		\ast&		&		\ast\\\hdashline
			\ast&		&		\ast&	&		\ast&		&		\ast&	\\
			&		\ast&		\lambda	&		\ast&		&		\ast&		&		\ast\\\hdashline
			\ast&		&		\ast&		&		\ast&		&		\ast&		\\
			&		\ast&	&		\ast&		\lambda	&		\ast&		&		\ast\\\hdashline
			\ast&		&		\ast&		&		\ast&		&		\ast&		\\
			&		\ast&		&		\ast&	\lambda	&		\ast&		&		\ast\\
			\end{array}
		\end{matrix} \right) , 
$$ if  $\deg\left( f_{\mathbf{R}}\left( \lambda \right) \right) =2$,  then there must exists  a  matrix $\mathbf{D}_4$ of the form  $\left( \begin{matrix}
	\begin{array}{cc:cc:cc}
		\ast&		&		1&		&		\ast&		\\
		&		\ast&		\ast&		\ast&		&		\ast\\\hdashline
		\ast&		&		&		&		\ast&		\\
		\ast&		&		&		&		\ast&		\\\hdashline
		\ast&		&		&		&		\ast&		\\
		&		\ast&		\ast&		\ast&		&		\ast\\
		\end{array}
	\end{matrix} \right) $  with   $\det \left( \mathbf{D}_4 \right)\ne 0 $. As in \textbf{Case 3}, we can obtain   $\det \left( \mathbf{D}_4 \right)= 0 $.

\end{document}